\documentclass[onecolumn,floatfix,aps,nofootinbib]{revtex4}
\pdfoutput=1
\usepackage[dvips]{epsfig}
\usepackage[english]{babel}
\usepackage[utf8]{inputenc}
\usepackage[T1]{fontenc}
\usepackage{hyperref}
\hypersetup{colorlinks=true, linkcolor=blue, citecolor=blue, urlcolor=blue}
\hypersetup{hidelinks,colorlinks=true,breaklinks=true,urlcolor= blue}
\hypersetup{%
    colorlinks = true,
    linkcolor  = blue,
    citecolor = cyan,
  }

\usepackage{lipsum}       
\usepackage{xcolor}       
\usepackage{bbm}
\usepackage{verbatim}
\usepackage{array}
\usepackage{bm} 
\usepackage{amsmath}
\usepackage{yfonts}
\usepackage{amsthm}
\usepackage{amsmath,amscd}
\usepackage{pst-plot}
\usepackage{slashed} 
\usepackage{tikz-cd}
\usepackage{amsfonts}
\usepackage{graphicx}
\usepackage{amssymb}

\usepackage{cases}
\usepackage[autostyle]{csquotes}  

\usepackage{tikz,tcolorbox}

\usepackage{tabularray}
\usepackage{mathtools}

\usepackage{indentfirst} 
\usepackage[titletoc]{appendix}
\usepackage{subeqnarray}
\usepackage{indentfirst} 

\usepackage{xcolor}
\usepackage{calrsfs}

\usepackage{enumitem}

\usepackage{wasysym}
\usepackage[all]{xy}
\usepackage{tikz-cd}

\usepackage{amsmath,amsfonts,amssymb,amsthm,mathrsfs,bbm,braket}

\newtheorem{defi}{Definition}[section]
\newtheorem{prop}{Proposition}[section]
\newtheorem{theo}{Theorem}[section]
\newtheorem{lemma}{Lemma}[section]

\numberwithin{equation}{section}

\usepackage{bbold}

\usepackage{hyperref}
\usepackage{setspace}

\usepackage{ulem}

\begin{document}

\title{Topologically induced deformations of differential forms}

\author{J. M. Hoff da Silva} 
\email{julio.hoff@unesp.br}
\affiliation{São Paulo State University (UNESP), School of Engineering and Sciences, Guaratinguetá - Brazil.}

\author{J. M. B. Matzenbacher} 
\email{j.matzenbacher@unesp.br}
\affiliation{São Paulo State University (UNESP), School of Engineering and Sciences, Guaratinguetá - Brazil.}

\author{R. da Rocha} 
\email{roldao.rocha@ufabc.edu.br}
\affiliation{Center of Mathematics, Federal University of ABC, 09210-580, Santo Andr\'e, Brazil.}

\begin{abstract}
Motivated by deformations of spin structures and their geometric implications, we study a deformed exterior derivative, which is a third-order nilpotent operator. This higher-order differential structure induces an anholonomic rescaling of the local frame and leads to a bifurcation of the associated cohomological structure, with two distinct notions of closedness and exactness. We prove that the Poincaré Lemma fails in both sectors, so that closed deformed forms need not be exact even locally. We then apply the resulting framework to a four-dimensional effective field theory and show that a constant saturation of the deformation parameter can generate a hierarchy between the electroweak and fundamental Planck scales through a purely geometric mechanism, without introducing extra spatial dimensions.
\end{abstract}

\maketitle

\section{Introduction}

Since the seminal work by Cartan \cite{car}, the interplay between spinors and spacetime geometry has been systematically investigated across various contexts \cite{pen}. To a large extent, these works treat spinor components as the building blocks of a pre-geometry, in the sense that spacetime coordinates can be expressed as bilinear products of spinor components. A complementary perspective arises from the existence of inequivalent spin structures on a given spacetime manifold.  Whenever $H^1(M,\mathbb{Z}_2)$ is nontrivial, distinct spin structures can exist over the same base manifold \cite{gp,milnor}. The resulting modification of the Dirac operator thus provides a local differential manifestation of global topological information.  For a review, see Ref. \cite{rev}. In particular, the heat-kernel coefficients and spectral properties of exotic Dirac operators have been shown to encode geometric invariants that distinguish the corresponding topological sectors \cite{DaRocha:2020oju,daRocha:2011yr}.

The physical relevance of this topological structure extends well beyond the construction of the Dirac operator itself. Quantum fields on multiply connected spacetimes have been investigated in a variety of contexts, including quantum field theory and finite-temperature systems \cite{is,Unwin:1979}, gravity and quantum gravity \cite{Hawking,Christensen}, superconductivity \cite{pp}, and many-body systems \cite{Boada}. Inequivalent spin structures have also been related to distinct effects in vacuum polarization and photon propagation \cite{Ford}, while their role has been explored in instanton compactification and confinement \cite{Sasaki,Schechter}. These results collectively demonstrate that the global topology encoded in the choice of spin structure can have concrete consequences for local field dynamics and physical observables \cite{HoffDaSilva:2020uov}.

In Ref. \cite{jhep}, an attempt was made to incorporate non-trivial topological terms into a metric formalism. Despite its interesting results, this approach led to a rather unwieldy framework. In addition, in Ref. \cite{jpa}, a specific deformation of spinor bundles that accounted for non-trivial topological effects was introduced, yielding an additional correction to the differential operator acting on spinor components. In the present work, we reformulate the foundations in Refs.  \cite{jhep,jpa} by incorporating the deformation directly into differential forms, demonstrating that the resulting forms are regular and nondegenerate in regions where the deformation function is nonvanishing and the resulting deformed coframe remains nondegenerate. To this goal, we systematically construct the apparatus of deformed differential forms originating from these spin-structure deformations. A central point of this reformulation is that the deformation is treated as part of the differential structure rather than as an additional term appended to the metric or to a field equation. This distinction allows the algebraic properties of the deformed differential operator to be analyzed independently of any particular physical model. The resulting differential calculus differs qualitatively from the ordinary de Rham calculus. The first nontrivial obstruction to nilpotency occurs at second order, while the third iterate vanishes. This distinction is structural rather than merely notational, since the failure of second-order nilpotency changes both the admissible notions of closedness and exactness and the homological algebra underlying the resulting complex. Our main mathematical result comprises the demonstration that the deformed exterior derivative $\tilde{d}$ satisfies a third-order nilpotency condition ($\tilde{d}^3 = 0$). An exterior derivative exhibiting such third-order nilpotency was previously studied in physics in Ref. \cite{ker} (albeit in a completely different context). In our framework, this higher-order nilpotency induces an anholonomic frame rescaling and a topological bifurcation in de Rham cohomology \cite{kap,vio,mr}. Within this structure, we prove that the Poincaré Lemma fails under both resulting cohomological definitions. Consequently, closed deformed forms are non-exact even locally, and precisely in regions where the underlying deformation is relevant, signaling the emergence of local topological obstructions. In this sense, the obstruction is intrinsic to the deformed differential calculus and should not be identified with the usual obstruction measured by the topology of the underlying manifold. Even on a contractible region, the deformed notion of exactness can therefore differ from its ordinary counterpart.

The geometric character of the deformation also makes it possible to investigate whether the same structure can modify physical scales. In particular, once the deformation is incorporated into the local frame, it induces a corresponding rescaling of the metric and of the kinetic and potential terms of fields defined on the deformed geometry. On the physical front, we apply this formalism to a four-dimensional effective field theory, demonstrating that a constant saturation of the deformation parameter naturally suppresses the electroweak scale relative to the fundamental Planck scale, thereby offering a purely geometric mechanism to mitigate the Higgs hierarchy problem in a manner analogous to that presented in Ref. \cite{rs}, but without resorting to extra spatial dimensions. At this stage, the construction should be understood as an effective mechanism for generating the required scale separation rather than as a complete dynamical solution of the hierarchy problem. In particular, the origin and stabilization of the saturation value of the deformation parameter remain to be established. The resulting picture therefore connects global information in the spin structure with local modifications of the differential calculus and, ultimately, with effective physical scales. The mathematical consistency of this chain, in particular the higher-order nilpotency and the associated failure of the local Poincaré Lemma, is the main focus of the analysis.

This paper is organized as follows. In Section \ref{sec2}, we begin with deformations of spinor bundles and arrive at the corresponding geometric deformations, thereby systematizing and generalizing the formulation presented in Ref. \cite{jpa} into a robust framework. In Section \ref{sec3}, we introduce the de Rham cohomological bifurcation and analyze the validity of the Poincaré Lemma. Section \ref{sec4} is devoted to addressing the Higgs hierarchy problem within this context. Finally, Section \ref{sec5} presents our concluding remarks and outlook.

\section{Deformations in Spinor bundles inducing geometrical deformations}
\label{sec2}
Consider a smooth, orientable, four-dimensional base manifold $M$ with vanishing second Stiefel--Whitney class, $w_2(M)=0$. Furthermore, assume that $M$ possesses a non-trivial fundamental group, $\pi_1(M)\neq 0$, and a non-trivial first \v{C}ech cohomology group with coefficients in $\mathbb{Z}_2$, denoted by $\check{H}^1(M,\mathbb{Z}_2)$. In this setting, multiple inequivalent spin structures for the same base manifold typically exist, and local mappings between sections of these different spin structures can be defined \cite{milnor}. For our argument, let $\rho_\alpha:U_\alpha\rightarrow U(1)$ be a function such that $\rho_\alpha^2=\rho_\beta^2\equiv \rho$ for all $x\in U_\alpha\cap U_\beta$. The function $\rho$ is globally continuous (even if the local functions $\rho_\alpha$ are not) and unimodular. This function induces a local mapping that connects spinors originating from sections of inequivalent bundles.

Let $U$ be an open set in a simple cover of $M$, and let $\varphi:U\rightarrow \mathbb{R}^*$ be a smooth, bounded, real-valued function. We define the map
\begin{eqnarray}\label{0}
i: U(1) &\to& \mathbb{C}^* \nonumber\\
\rho(x) &\mapsto& (\varphi\rho)(x) = \varphi(x)\rho(x),
\end{eqnarray}
with squared modulus $\varphi^2$, as the $\varphi$-deformation of maps between spin structure sections. Let $P_1$ and $P_2$ be two spin structures over $M$ with respective transition functions $\gamma^1_{\alpha\beta}, \gamma^2_{\alpha\beta}: U_\alpha\cap U_\beta\rightarrow \operatorname{Spin}(1,3)$, and define the maps $\xi_{\alpha\beta}\equiv \gamma^1_{\alpha\beta}\gamma^2_{\beta\alpha}: U_\alpha\cap U_\beta\rightarrow \operatorname{Spin}(1,3)$, whose nontriviality reflects the inequivalence of $P_1$ and $P_2$. Before clarifying the role played by the $\varphi$-deformation in this context, we note that $\xi_{\alpha\beta}$ determines an element $[\xi] \in \check{H}^1(M,\mathbb{Z}_2)$, representing the difference class of the spin structures. For instance, a standard result applies: if $P_1\simeq P_2$, then $[\xi]=e$. Indeed, if $P_1\simeq P_2$, there exists a bundle isomorphism $\phi:P_1\rightarrow P_2$. Let $\{\lambda^1_\alpha: U_\alpha\rightarrow P_1\}$ be a system of local sections for $P_1$, and define $\{\lambda^2_\alpha\}$ similarly for $P_2$, so that $\lambda^2_\alpha=\phi\lambda^1_\alpha$. Applying $\phi$ to the local transformation law $\lambda^1_\sigma=\lambda^1_\alpha\gamma^1_{\alpha\sigma}$, directly yields $\lambda^2_\sigma(\lambda^2_\alpha)^{-1}=\gamma^1_{\alpha\sigma}$. Hence, $\gamma^1_{\alpha\sigma}=\gamma^2_{\alpha\sigma}$, which implies $[\xi]=e$. We are now in a position to analyze a non-standard result in the context of $\varphi$-deformations.
\begin{prop}
Let $P_1$ and $P_2$ be two spin structures over the same base manifold $M$. If the difference class $[\xi]$ is trivial, then $P_1\simeq P_2$ if and only if $\varphi(x)$ is a globally constant function\footnote{The equivalence considered in the following proposition is not the ordinary equivalence of spin structures, which is determined solely by the underlying principal bundles and their transition functions. Rather, it concerns the existence of a globally well-defined map compatible with the additional $\varphi$-deformation introduced in Eq. \eqref{0}. The condition on $\varphi$ therefore characterizes compatibility with the deformed construction, rather than ordinary isomorphism of the undeformed spin structures.}.
\end{prop}
\begin{proof}
    Since $[\xi]=e$, there exists a 0-cochain $q$ (with values in $\mathbb{Z}_2$) such that $\xi_{\alpha\beta}=q(\alpha)q^{-1}(\beta)$ with a trivial coboundary, i.e., $\delta\xi_{\alpha\beta}=e$. Therefore, $\gamma_{\alpha\beta}^1\gamma_{\beta\alpha}^2=q(\alpha)q^{-1}(\beta)$, yielding
\begin{equation}\label{1}
\gamma^2_{\alpha\beta}=q^{-1}(\alpha)\gamma^1_{\alpha\beta}q(\beta), \quad x\in U_\alpha\cap U_\beta.
\end{equation}
Let $\pi:TM\rightarrow M$ be the canonical projection, and define $L_\alpha \coloneqq \pi^{-1}(U_\alpha)$ and $\gamma_\alpha:U_\alpha\rightarrow \operatorname{Spin}(1,3)$. Observe that $\gamma(v)\in \operatorname{Spin}(1,3)$ is the unique element satisfying $v=\lambda_\alpha^1(\pi(v))\gamma_\alpha(v)$. According to Eq. \eqref{0}, a map $\phi$ between the spin structures $P_1$ and $P_2$ defined over open sets covering $M$ must incorporate the $\varphi$-deformation to account for this new degree of freedom. Based on this reasoning, we define $\tilde{L}_\alpha \coloneqq \mathbb{R}^*\times L_\alpha$, which induces an extension of the structure group to $\operatorname{Spin}(1,3)\times \mathbb{R}^+$ given by pairs $(k,a)$ with $k\in \operatorname{Spin}(1,3)$ and $a\in\mathbb{R}^+$. The extension by $\mathbb R^+$ should be understood as an auxiliary enlargement of the structure group that records the local scale degree of freedom introduced by $\varphi$. It does not replace the underlying spin structure. Rather, it provides the bundle-theoretic setting in which the spinorial and scaling transformations can be treated simultaneously.
    
    Let $\phi_\alpha:\tilde{L}_\alpha\rightarrow P_2$ be a map equivariant under $\operatorname{Spin}(1,3)\times \mathbb{R}^+$ transformations. Thus, imposing homogeneous propagation under the rescaling $\varphi\mapsto l\varphi$ implies that $\phi_\alpha$ depends on $\varphi_\alpha$ as $\varphi_\alpha^n$ for a fixed exponent $n$. This observation motivates defining $\phi_\alpha$ as
\begin{equation}\label{2}
\phi_\alpha = C\varphi_\alpha^n(\pi(v))\lambda^2_\alpha(\pi(v))q(\alpha)\gamma_\alpha(v),
\end{equation}
where $C$ is a constant. Since $q\in\mathbb{Z}_2$, the center of $\operatorname{Spin}(1,3)$, Eq. \eqref{1} enables us to write
\begin{equation}\label{3}
\phi_\beta = C\varphi_\beta^n(\pi(v))\lambda_\alpha^2(\pi(v))q(\alpha)\gamma^1_{\alpha\beta}\gamma_\beta(v)
\end{equation}
for $v\in \tilde{L}_\alpha\cap\tilde{L}_\beta$. Recalling that $v=\lambda_\beta^1(\pi(v))\gamma_\beta(v)$, a straightforward calculation yields $[\lambda^1_\alpha(\pi(v))]^{-1}v=\gamma^1_{\alpha\beta}\gamma_\beta(v)$. Hence, using $v=\lambda_\alpha^1(\pi(v))\gamma_\alpha(v)$, we obtain $\gamma_\alpha(v)=\gamma^1_{\alpha\beta}\gamma_\beta(v)$, from which Eq. \eqref{3} gives
\begin{equation}\label{4}
\phi_\beta = \left(\frac{\varphi_\beta}{\varphi_\alpha}(\pi(v))\right)^n \phi_\alpha.
\end{equation}
After \eqref{4}, it is straightforward to see that $\phi$ defines a well-defined global map if and only if $\varphi$ is constant. 
\end{proof} 

The above result asserts that in regions where the $\varphi$-deformation remains relevant, even if simply connected, distinct spin structures can coexist, provided they exist initially. In other words, when $\check{H}^1(M,\mathbb{Z}_2)$ is nontrivial, there exist as many distinct spin structures as there are elements in $\check{H}^1(M,\mathbb{Z}_2)$. The next result is a direct adaptation that settles this point conclusively.
\begin{theo}
Let $P=(M,\pi,\operatorname{Spin}(1,3)\times \mathbb{R}^+)$ be a spin structure with transition functions $\gamma_{\alpha\beta}$, let $g_{\alpha\beta}$ be the transition functions for $B=(M,\pi_B,\operatorname{SO}(1,3)\times\mathbb{R}^+)$, let $\bar{\rho}:\operatorname{Spin}(1,3)\times\mathbb{R}^+\rightarrow \operatorname{SO}(1,3)\times\mathbb{R}^+$ be a homomorphism such that $\bar{\rho}(\gamma_{\alpha\beta})=g_{\alpha\beta}$, and let $[\alpha]\in \check{H}^1(M,\mathbb{Z}_2)$. Then, there exists a spin structure $P'$ such that the difference class between $P$ and $P'$ is $[\alpha]$.
\end{theo}
\begin{proof}
   Let $\varphi_{\beta\theta}:U_\beta\cap U_\theta\to\mathbb{R}^+,$ with 
$\displaystyle\varphi_{\beta\theta}(x)
=
\frac{\varphi_\theta(x)}{\varphi_\beta(x)}.$ 
The corresponding cocycle condition is satisfied, since
\begin{equation}
\varphi_{\beta\theta}
\varphi_{\mu\theta}^{-1}
\varphi_{\mu\beta}
=
\frac{\varphi_\theta}{\varphi_\beta}
\frac{\varphi_\mu}{\varphi_\theta}
\frac{\varphi_\beta}{\varphi_\mu}
=1.
\end{equation} Define $\gamma'_{\beta\theta}=\alpha^{-1}_{\beta\theta}\gamma_{\beta\theta}$, where $\alpha_{\beta\theta}$ is a 1-cocycle taking values in $\mathbb{Z}_2$. Thus,\begin{equation}(\gamma'_{\beta\theta},\varphi_{{\beta\theta}})(\gamma'_{\mu\theta},\varphi_{{\mu\theta}})^{-1}(\gamma'_{\mu\beta},\varphi_{{\mu\beta}}) = (\gamma'_{\beta\theta}\gamma'^{-1}_{\mu\theta}\gamma'_{\mu\beta},\varphi_{{\beta\theta}}\varphi_{{\mu\theta}}^{-1}\varphi_{{\mu\beta}}),\end{equation} which yields\begin{equation}\label{d2}(\gamma'_{\beta\theta},\varphi_{{\beta\theta}})(\gamma'_{\mu\theta},\varphi_{{\mu\theta}})^{-1}(\gamma'_{\mu\beta},\varphi_{{\mu\beta}}) = ((\alpha_{\beta\theta}\alpha^{-1}_{\mu\theta}\alpha_{\mu\beta})^{-1}\gamma_{\beta\theta}\gamma^{-1}_{\mu\theta}\gamma_{\mu\beta},\varphi_{{\beta\theta}}\varphi_{{\mu\theta}}^{-1}\varphi_{{\mu\beta}}).\end{equation}Since $\alpha_{\mu\beta}$ is a cocycle, \eqref{d2} evaluates to $(e,1)$, implying the existence of a principal bundle $P'$ whose transition functions are $\gamma'_{\beta\theta}$. Moreover, $\ker\bar{\rho}=\{1,-1\}\cong\mathbb{Z}_2\subset Z(\operatorname{Spin}(1,3)\times \mathbb{R}^+)=\mathbb{Z}_2\times\mathbb{R}^+$, showing that $\bar{\rho}$ is a central homomorphism. Standard results then apply to ensure the existence of an equivariant map from $P'$ to $B$, confirming that $P'$ is a spin structure.
\end{proof}

The impact of the $\varphi\rho$ mapping is directly reflected in the derivative operator acting on sections of the additional (exotic) spin structures. The contribution arising from the standard map $\rho$ to the derivative operator was derived in several works \cite{is,pp,rev}, while the correction due to the composite map $\varphi\rho$ was established in \cite{jpa}. Here, we frame the analysis within the context established in the proof of the previous proposition: a standard term $\frac{1}{2}\rho^{-1}\partial\rho$ arising from the $\mathfrak{spin}(1,3)$ algebra, alongside an additional term $\varphi^{-1}\partial\varphi$ coming from the scaling algebra. The net result is given by the replacement
\begin{equation}\label{d1}
\partial_i \mapsto \partial_i + \frac{1}{2}\rho^{-1}\partial_i\rho + \varphi^{-1}\partial_i\varphi, \quad i=0, 1,2,3,
\end{equation}
whenever an exotic spinor component is differentiated. Every exotic spinor component is differentiated according to \eqref{d1}. The construction thus separates two logically distinct pieces of information.
The nontrivial topology is encoded globally through the inequivalence of
spin structures and the associated transition data, whereas the function
$\varphi$ provides the local deformation through which this information
enters the differential calculus. The resulting local deformation can
therefore remain nontrivial even in a region whose ordinary topology is
trivial.

The final ingredient required for a nontrivial topology to induce specific geometric deformations is the connection established by Cartan between spinors and spacetime points \cite{car}. As demonstrated in Refs. \cite{car,pen}, any given spacetime point $(x^0,\dots,x^3)$ can be written in terms of a product of spinor components. Generically, one writes $x^i\sim (\zeta^*\chi)^i$, where $\zeta$ and $\chi$ denote spinor components \cite{pen,ox}. In this way, assuming the existence of a local Cartesian coordinate system $\{x^i\}$ on an open subset $U\subset M$ and a smooth function $f:U\rightarrow \mathbb{R}$, the standard differential $df$ expands at first order as 
\begin{equation}
df = \sum_{i=1}^{\dim M}\frac{\partial f}{\partial x^i}dx^i = \sum_{i,j=1}^{\dim M}\frac{\partial f}{\partial x^j}\frac{\partial x^j}{\partial x^i}dx^i \sim \sum_{i,j=1}^{\dim M}\frac{\partial f}{\partial x^j}\frac{\partial ({\zeta^*\chi})^j}{\partial x^i}dx^i.
\end{equation}
In the setup of the replacement \eqref{d1}, it is natural to define the mapping:
\begin{equation}
    df\mapsto \sum_{i,j=1}^{\dim M}\frac{\partial f}{\partial x^j}\Bigg[\frac{\partial}{\partial x^i}+\rho^{-1}\frac{\partial \rho}{\partial x^i}+2\varphi^{-1}\frac{\partial \varphi}{\partial x^i}\Bigg](\zeta^*\chi)^jdx^i, \label{dd2}
\end{equation} where \eqref{d1} is applied to each spinor component $x^i\sim(\zeta^*\chi)^i$, hereafter associated with exotic spinors. Usually, in the absence of group torsion, one can set $\rho=\exp(i\xi(x))$ \cite{pp,rev,jpa,daSilva:2023qnx} for a smooth real-valued function $\xi(x)$. Hence, rewriting the first term on the right-hand side of \eqref{dd2} back as a differential, we arrive at
\begin{equation}
df \mapsto df + \sum_{i=1}^{\dim M}\left[i\,d\xi\,\frac{\partial f}{\partial x^i}x^i + d(\ln \varphi^2)\frac{\partial f}{\partial x^i}x^i\right].
\end{equation}

We adopt a partition of the base manifold following the scheme detailed in \cite{jpa}. Let $\widetilde{M}\subset M$ be partitioned as $\widetilde{M}=\Sigma \cup \widetilde{\mathbb{R}}^n \cup (\widetilde{M}\setminus (\Sigma \cup \widetilde{\mathbb{R}}^n))\cong\mathbb{R}^n$, where $\pi_1(\Sigma)\neq 0$ and $n=\dim M$. Here, $\widetilde{\mathbb{R}}^n$ is simply connected and $\xi(x)$ is constant for all $x\in \widetilde{\mathbb{R}}^n$, though the local imprint of the deformation associated with the
nontrivial spin structure is encoded there through the varying
$\varphi(x)$. While Section \ref{sec4} explores an application of this formalism in $\mathbb{R}^n$, the present analysis focuses on the consequences in $\widetilde{\mathbb{R}}^n$, where
\begin{equation}\label{d3}
df|_{\widetilde{\mathbb{R}}^n}\mapsto df+\sum_{i=1}^n d(\ln\varphi^2)\frac{\partial f}{\partial x^i}x^i.
\end{equation}
To find a basis for the dual space of $\widetilde{\mathbb{R}}^n$, we now specialize $f$ to the coordinate functions $x^k$, exploiting the linear nature of our implementation of Cartan's formalism. Expression \eqref{d3} then reduces to 
\begin{equation}\label{reti1}
dx^k\vert{}_{\widetilde{\mathbb{R}}^n}\mapsto dx^k+x^k d(\ln\varphi^2).
\end{equation}
The deformation obtained so far can be recast in the form presented in Ref. \cite{jhep}. While it yields interesting physical consequences, directly handling the additional term presents technical challenges. For instance, a standard Fourier transform would have an overly restrictive domain of validity and require several extra considerations to avoid infinities. To circumvent this issue, we note the identity $x^k\partial_j(\ln\varphi^2)=\partial_j[x^k \ln\varphi^2]-\delta_j^k \ln\varphi^2$, which allows the transformation (\ref{reti1}) to be rewritten as
\begin{equation}\label{ti}
dx^k|_{\widetilde{\mathbb{R}}^n}\mapsto \tilde{\tilde{d}}x^k \coloneqq (1-\ln\varphi^2)dx^k + d(x^k\ln\varphi^2).
\end{equation} Therefore, by subtracting the last term of $\tilde{\tilde{d}}x^k$ (in a procedure analogous to subtracting counterterms to absorb infinities in quantum field theory), that is, considering $\tilde{\tilde{d}}x^k - d(x^k\ln\varphi^2) \coloneqq \tilde{d}x^k = (1-\ln\varphi^2)dx^k$, we are left with a regular and manageable deformed basis for $(\widetilde{\mathbb{R}}^n)^*$. It is indeed straightforward to verify that the set $\{\tilde{d}x^i = (1-\theta(x))dx^i\}$, where $\theta(x)\equiv\ln\varphi^2$, is linearly independent and spans the dual space of $\widetilde{\mathbb{R}}^n$. A few remarks are in order at this point. First, to avoid ill-defined behavior, suitable boundary conditions must be imposed on $\theta(x)$ (or, equivalently, on $\varphi(x)$) within the region $\widetilde{\mathbb{R}}^n$ and across the transitions to $\Sigma$ and $\mathbb{R}^n$. In particular, in Section \ref{sec4}, we explore a saturation value for $\varphi$ that allows us to address the hierarchy problem. Furthermore, although the formalism developed here might seem to require a global Cartesian coordinate system for generalization to the entire manifold (especially in applying Cartan's formalism), standard constructions such as Riemannian normal coordinates suffice to resolve this issue.  

\section{de Rham cohomological bifurcation}
\label{sec3}

    We now explore the main properties of differential forms on $\widetilde{\mathbb{R}}^n$, where the differential operator is modified by the nontrivial topology, thus establishing a suitable framework for investigating Poincaré Lemma. Let $\Lambda_1(\widetilde{\mathbb{R}}^n)$ and $\Lambda^1(\widetilde{\mathbb{R}}^n)$ denote $\widetilde{\mathbb{R}}^n$ (as a vector space) and its dual $(\widetilde{\mathbb{R}}^n)^*$, respectively. Accordingly, $\Lambda^k(\widetilde{\mathbb{R}}^n)$ denotes the space of differential $k$-forms, while $\Lambda_k(\widetilde{\mathbb{R}}^n)$ denotes the space of $k$-vectors. As previously discussed, the relation between the deformed exterior algebra basis and the standard basis is $\tilde{d}x^i = (1-\theta(x))\,dx^i$. Therefore, an arbitrary $1$-form $\omega_{(1)}\in\Lambda^1(\widetilde{\mathbb{R}}^n)$, given by $\omega_{(1)} = \omega_i(x)\,\tilde{d}x^i$, can be written in the standard basis as\begin{equation}\omega_{(1)} = (1-\theta(x))\omega_i(x)\,dx^i.\end{equation}By the same reasoning, a generic $2$-form $\omega_{(2)}\in \Lambda^2(\widetilde{\mathbb{R}}^n)$ acquires a factor of $(1-\theta(x))^2$, taking the general form $\omega_{(2)} = (1-\theta(x))^2 \omega_{ij}(x)\,dx^i\wedge dx^j$. In both expressions, $\omega_i(x)$ and $\omega_{ij}(x)$ denote the corresponding coefficient functions. Proceeding inductively, this pattern holds for all higher-degree forms; thus, a general $k$-form $\omega \in \Lambda^k(\widetilde{\mathbb{R}}^n)$ is represented as\begin{equation}\omega = (1-\theta(x))^k \omega_{i_1 i_2\dots i_k}(x) \,dx^{i_1}\wedge dx^{i_2}\wedge \dots \wedge dx^{i_k}.\end{equation}
 \begin{defi}
    The deformed exterior derivative of a $k$-form is the map $\tilde{d}: \Lambda^k(\widetilde{\mathbb{R}}^n) \to \Lambda^{k+1}(\widetilde{\mathbb{R}}^n)$ defined by
\begin{equation}\label{cadeira}
\tilde{d}\omega = \partial_n\left[ (1-\theta(x))^k \omega_{i_1 i_2\dots i_k}(x)\right]\,\tilde{d}x^n\wedge dx^{i_1}\wedge dx^{i_2}\wedge \dots \wedge dx^{i_k}.
\end{equation}
\end{defi} We are now in a position to state the result upon which the cohomological bifurcation rests.
    \begin{prop}
        \label{propnilpotencia}
        On $\widetilde{\mathbb{R}}^n$, a deformed differential $k$-form is nilpotent of order three with respect to the deformed exterior derivative.
    \end{prop}
        \begin{proof}
            Note that expression \eqref{cadeira} can be rewritten as
\begin{equation}
\tilde{d}\omega = (1-\theta(x))^k\left[(1-\theta(x))\partial_n\omega_{i_1 i_2\dots i_k}(x) - k\omega_{i_1 i_2\dots i_k}(x)\partial_n\theta(x) \right]dx^n\wedge dx^{i_1}\wedge dx^{i_2}\wedge \dots \wedge dx^{i_k}.
\end{equation}
A direct calculation shows that the second exterior derivative of $\omega$ reads
\begin{equation}
\tilde{d}^2\omega = -(1-\theta(x))^{k+1}\partial_m\theta(x)\partial_n\omega_{i_1 i_2\dots i_k}(x)\,dx^m\wedge dx^n\wedge dx^{i_1}\wedge dx^{i_2}\wedge \dots \wedge dx^{i_k},
\end{equation}
from which applying a third derivative yields $\tilde{d}^3\omega = 0$ identically.
        \end{proof}

It is important to stress that the relation $\tilde d^{\,3}=0$ is understood as an identity on the deformed differential algebra defined above, rather than as a property of the ordinary exterior derivative acting on the undeformed basis. In particular, the non-vanishing of $\tilde d^{\,2}$ is an intrinsic feature of the deformed calculus and is not attributable to a failure of the ordinary exterior derivative itself. The distinction between the ordinary and deformed calculi will be maintained throughout the subsequent discussion.

Let $\mathring\Omega(M)$ denote the graded space of differential forms on $M$, and suppose that
\begin{equation}
\tilde{d}:\Omega^k(M)\rightarrow\Omega^{k+1}(M),
\qquad
\tilde{d}^{\,3}=0.
\end{equation}
Thus, $\tilde{d}$ does not define, in general, an ordinary
differential complex, since $\tilde{d}^{\,2}\neq0$. Rather, the
relation $\tilde{d}^{\,3}=0$ places the construction in the
framework of a $3$-complex (or $N$-complex with $N=3$), provided the
graded structure and the degree-one action are understood with respect
to the deformed differential algebra defined above. First, define
\begin{equation}
H_{(1)}^k(\mathring\Omega(M),\tilde{d})
:=
\frac{
\ker\!\left(
\tilde{d}:\Omega^k(M)\rightarrow\Omega^{k+1}(M)
\right)
}{
\operatorname{im}\!\left(
\tilde{d}^{\,2}:\Omega^{k-2}(M)\rightarrow\Omega^k(M)
\right)
}.\label{quo}
\end{equation}
The quotient (\ref{quo}) is well defined, since $\tilde{d}^{\,3}=0$ implies
$\operatorname{im}\tilde{d}^{\,2}
\subseteq
\ker\tilde{d},$ as 
$\tilde{d}(\tilde{d}^{\,2}\alpha)
=
\tilde{d}^{\,3}\alpha
=
0,$ where $\alpha$ is an arbitrary $k$-form.  
Second, define
\begin{equation}
H_{(2)}^k(\mathring\Omega(M),\tilde{d})
:=
\frac{
\ker\!\left(
\tilde{d}^{\,2}:\Omega^k(M)\rightarrow\Omega^{k+2}(M)
\right)
}{
\operatorname{im}\!\left(
\tilde{d}:\Omega^{k-1}(M)\rightarrow\Omega^k(M)
\right)
}.
\end{equation}
This quotient is likewise well defined, since
$
\operatorname{im}\tilde{d}
\subseteq
\ker\tilde{d}^{\,2},
$ as a consequence of
the relation $\tilde{d}^{\,2}(\tilde{d}\alpha)
=
\tilde{d}^{\,3}\alpha
=
0.$ 
Thus, the two generalized cohomology sectors are respectively represented by 
\begin{equation}
{
H_{(1)}^k(\tilde{d})
=
\frac{
\ker\tilde{d}|_{\Omega^k}
}{
\operatorname{im}\tilde{d}^{\,2}|_{\Omega^{k-2}}
}
},\qquad\quad{
H_{(2)}^k(\tilde{d})
=
\frac{
\ker\tilde{d}^{\,2}|_{\Omega^k}
}{
\operatorname{im}\tilde{d}|_{\Omega^{k-1}}
}
}.
\end{equation}
In the ordinary de Rham  case, where $\tilde{d}^{\,2}=0$, the first sector reduces to the usual de Rham (dR) cohomology,
\begin{equation}
H_{(1)}^k(\tilde{d})
=
\frac{\ker\tilde{d}}
{\operatorname{im}\tilde{d}}
\equiv
H_{\mathrm{dR}}^k(M).
\end{equation}
Therefore, when $\tilde{d}^{\,2}\neq0$ but $\tilde{d}^{\,3}=0$, the appropriate mathematical structure is not an ordinary de Rham complex, but rather a $3$-complex with the two generalized cohomology sectors $H_{(1)}^k(\tilde{d})$ and $H_{(2)}^k(\tilde{d})$. In this sense, the transition from the ordinary de Rham regime $\tilde{d}^{\,2}=0$ to the higher-order regime $\tilde{d}^{\,2}\neq0$ with $\tilde{d}^{\,3}=0$ can be characterized as a {cohomological bifurcation}, provided that this terminology is explicitly defined in terms of the corresponding change in the cohomological structure.

The construction is closely related in spirit to the generalized cohomology associated with $N$-complexes, although the present realization is geometrically motivated by a deformation of spin structures and of the associated differential basis. We therefore do not assume that the general algebraic results for abstract $N$-complexes automatically apply to the present differential calculus. Establishing the precise correspondence between the two frameworks is an interesting question in its own right. It is worth emphasizing that we use the term cohomological bifurcation here in a descriptive sense. The ordinary single de Rham cohomology is replaced, once second-order nilpotency is lost, by two naturally associated generalized cohomological sectors. This terminology does not imply that these generalized groups are themselves de Rham cohomology groups in the standard sense.
       
    As a result, two distinct ways of defining closedness and exactness for deformed differential forms emerge. This underlies the bifurcated structure of the topology (see \cite{kap,vio} for generalizations). Given a deformed differential $k$-form $\omega \in \Lambda^k(\widetilde{\mathbb{R}}^n)$, there are two natural approaches to the fundamental concepts of closedness and exactness:
\begin{enumerate}
\item $\omega$ is said to be closed if $\tilde{d}\omega = 0$, and exact if there exists a deformed differential $(k-2)$-form $\eta \in \Lambda^{k-2}(\widetilde{\mathbb{R}}^n)$ such that $\omega = \tilde{d}^2\eta$;
\item alternatively, $\omega$ is said to be closed if $\tilde{d}^2\omega = 0$, and exact if there exists a deformed differential $(k-1)$-form $\eta \in \Lambda^{k-1}(\widetilde{\mathbb{R}}^n)$ such that $\omega = \tilde{d}\eta$.
\end{enumerate}
    It can be readily verified that, in both cases, every exact form is closed, just as in the standard formalism. Nevertheless, the converse requires a more detailed analysis. Determining the conditions under which a closed form is exact is precisely what motivates the study of the Poincaré Lemma in the present context. We now examine how the standard construction applies to the deformed basis configuration. Let the natural inclusion $i_t$ be defined as
        \begin{eqnarray}
            i_t: \widetilde{A}\subset\widetilde{\mathbb{R}}^n &\hookrightarrow& \mathbb{R}\times \widetilde{\mathbb
        {R}}^n\nonumber\\
        (x^{i_1},x^{i_2},\cdots,x^{i_n}) &\mapsto& i_t(x^{i_1},x^{i_2},\cdots,x^{i_n}) = (t,x^{i_1},x^{i_2},\cdots,x^{i_n}),
        \end{eqnarray}
        where $\widetilde{A}$ is an open subset of $\widetilde{\mathbb{R}}^n$. The map $i_t$ induces a pullback $i^*_t$ given by
    \begin{eqnarray}
        i^*_t : \Lambda^k(\mathbb{R}\times \widetilde{\mathbb{R}}^n) &\hookrightarrow& \Lambda^k(\widetilde{\mathbb{R}}^n)\nonumber\\
        \omega &\mapsto& \omega_t \coloneq i^*_t\omega = \omega\circ i_t,
    \end{eqnarray}
    where $\omega \in \Lambda^k(\mathbb{R}\times \widetilde{\mathbb{R}}^n)$. Here $i_t^*$ denotes the ordinary pullback of the underlying differential forms. It should not be assumed a priori that this pullback intertwines $\tilde d$ with the corresponding deformed differential on $\widetilde{\mathbb R}^n$. Whenever such an identity is required below, it is established explicitly from the definitions. Whenever $\tilde{d}$ acts on a
pulled-back form, its action is computed directly from the definition
\eqref{cadeira}. The following definition and lemma mirror the standard construction with slight modifications, but we include them to keep the presentation self-contained. A crucial difference from the standard proof of the Poincaré Lemma is that the deformed differential is not, in general, natural with respect to arbitrary pullbacks. Consequently, the usual homotopy argument cannot simply be imported from the ordinary de Rham complex. The purpose of the construction below is therefore not to assume the standard homotopy identity, but to determine explicitly which additional terms arise from the deformation and whether they can vanish under the present definition of $\tilde d$.
    
\begin{defi}
Let $\widetilde{A}$ be a subset of $\widetilde{\mathbb{R}}^n$. The homotopy operator $H$ is the linear map
\begin{equation}
H : \Lambda^{k+1}(\mathbb{R}\times \widetilde{A}) \to \Lambda^{k}(\widetilde{A})
\end{equation}
such that, for any $\omega \in \Lambda^{k+1}(\mathbb{R}\times \widetilde{A})$, one has:
\begin{enumerate}[label=\roman*)]
\item $H\omega = 0$, if $\omega = (1-\theta(t,x))^{k+1}\omega_{i_1 i_2\dots i_{k+1}}(t,x)\,\, dx^{i_1}\wedge dx^{i_2}\wedge \dots \wedge dx^{i_{k+1}}$;
\item $H\omega = \left(\int_0^1 (1-\theta(t,x))^{k+1}\omega_{t i_1\dots i_k}(t,x) \, dt\right)\,\, dx^{i_1}\wedge dx^{i_2}\wedge \dots \wedge dx^{i_k}$, if 
\begin{equation*}
\omega = (1-\theta(t,x))^{k+1}\omega_{t i_1\dots i_k}(t,x)\,\, dt\wedge dx^{i_1}\wedge dx^{i_2}\wedge \dots \wedge dx^{i_k}.
\end{equation*}
\end{enumerate}
\end{defi}

    \begin{lemma}\label{leminha}
The following properties hold for the homotopy operator $H$:
\begin{enumerate}[label=\roman*)]
\item It commutes with $C^\infty$ functions that do not depend on $t$;
\item It satisfies $H(\tilde{d}\omega) + \tilde{d}(H\omega) = \omega_1 - \omega_0$ for all $\omega \in \Lambda^{k+1}(\mathbb{R}\times \widetilde{A})$, where $\omega_1 = i^*_1\omega$ and $\omega_0 = i^*_0\omega$.
\end{enumerate}
\end{lemma}
    \begin{proof}
        Consider a deformed differential $(k+1)$-form $\omega$ and a $t$-independent function $f(x)$. Item (i) in Lemma \ref{leminha} follows trivially from the fact that $f(x)$ can be pulled outside the integral, yielding $[H,f(x)]\omega = 0$. Item (ii) in Lemma \ref{leminha} is proven by splitting the argument into two cases. In the first case, assume that $\omega$ does not involve $dt$. Then, by definition, $\tilde{d}(H\omega) = 0$. Defining $\alpha(t,x) = (1-\theta(t,x))^{k+1}\omega_{i_1 i_2\dots i_{k+1}}(t,x)$, the computation of $H(\tilde{d}\omega)$ yields\begin{equation}H(\tilde{d}\omega) = H\left[\left(\frac{\partial \alpha(t,x)}{\partial t }dt + \frac{\partial \alpha(t,x)}{\partial x^j }\tilde{d}x^j\right)\wedge dx^{i_1}\wedge \dots \wedge dx^{i_{k+1}} \right] = \left(\int_0^1\frac{\partial \alpha(t,x)}{\partial t }dt\right)\,dx^{i_1}\wedge \dots \wedge dx^{i_{k+1}}.\end{equation}Therefore, $H(\tilde{d}\omega) = (\alpha(1,x)-\alpha(0,x))\,dx^{i_1}\wedge \dots \wedge dx^{i_{k+1}} = \omega_1 - \omega_0$, from which $H(\tilde{d}\omega) + \tilde{d}(H\omega) = \omega_1 - \omega_0$. In the second case, assume that $\omega$ involves $dt$. Setting $\gamma(t,x) = (1-\theta(t,x))^{k+1}\omega_{t i_1\dots i_k}(t,x)$, one finds\begin{equation}H(\tilde{d}\omega) = -\left(\int_0^1 \frac{\partial \gamma (t,x)}{\partial x^j}dt\right)\,\tilde{d}x^j\wedge dx^{i_1}\wedge \dots \wedge dx^{i_k}.\end{equation}On the other hand, a similar computation shows that $\tilde{d}(H\omega) = -H(\tilde{d}\omega)$, yielding $H(\tilde{d}\omega) + \tilde{d}(H\omega) = 0$. It is worth noting that for forms containing $dt$, the pullback $i^*_t\omega$ vanishes, which implies $\omega_1 - \omega_0 = 0$.
    \end{proof}

Furthermore, under Scenario 1 above (in which a form $\omega \in \Lambda^{k}(\widetilde{\mathbb{R}}^n)$ is defined to be closed when $\tilde{d}\omega = 0$), let $\widetilde{A}\subset \widetilde{\mathbb{R}}^n$ be a star-shaped domain. Define the dilation map $\mathcal{H}$ as
\begin{eqnarray}
    \mathcal{H}:\mathbb{R}\times \widetilde{A} &\to& \widetilde{A},\nonumber \\
(t,x^1,x^2,\dots, x^n) &\mapsto& \mathcal{H}(t,x^1,x^2,\dots, x^n) = (tx^1, tx^2, \dots, tx^n).
\end{eqnarray}
Note that $\mathcal{H}\circ i_0 : \widetilde{A}\to \{0\}$ and $\mathcal{H}\circ i_1 = \text{id}_{\widetilde{A}}$. From Lemma \ref{leminha}, one obtains
\begin{equation}
\tilde{d}(H(\mathcal{H}^*\omega)) + H(\tilde{d}(\mathcal{H}^*\omega)) = i_1^*(\mathcal{H}^*\omega) - i_0^*(\mathcal{H}^*\omega).
\end{equation}
Since $i_1^*(\mathcal{H}^*\omega) = (\mathcal{H}\circ i_1)^*\omega = \omega$ and $i_0^*(\mathcal{H}^*\omega) = (\mathcal{H}\circ i_0)^*\omega = 0$ (for $k \ge 1$), it follows that
\begin{equation}
\tilde{d}(H(\mathcal{H}^*\omega)) + H(\tilde{d}(\mathcal{H}^*\omega)) = \omega.
\end{equation}
Assuming $\omega$ to be a closed deformed form ($\tilde{d}\omega = 0$), we have $\tilde{d}(\mathcal{H}^*\omega) = 0$, which yields $\tilde{d}(H(\mathcal{H}^*\omega)) = \omega$. Let us examine this last relation in detail. Suppose, for instance, that $\omega$ is exact. Then, there exists a differential $(k-2)$-form $\eta$ such that $\omega = \tilde{d}^2\eta$, and thus $\tilde{d}(H(\mathcal{H}^*\omega)) = \tilde{d}^2\eta$. Owing to the third-order nilpotency property ($\tilde{d}^3 = 0$), applying the exterior derivative to this last expression leads to $\tilde{d}^2(H(\mathcal{H}^*\omega)) = 0$.
        
        The analysis of $\tilde{d}^2(H(\mathcal{H}^*\omega))$ can be split into two distinct cases. If $\mathcal{H}^*\omega$ does not involve $dt$, then $H(\mathcal{H}^*\omega)=0$ and $\tilde{d}^2(H(\mathcal{H}^*\omega)) = 0$ holds trivially without additional constraints. In the second case, where $\mathcal{H}^*\omega$ involves $dt$, the condition $\tilde{d}^2(H(\mathcal{H}^*\omega)) = 0$ implies
        \begin{eqnarray}\int_0^1 (1-\theta(tx))^k\left[\frac{k}{(1-\theta(tx))}\partial_{i_2}\theta(tx)\omega_{t i_3\dots i_{k+1}}(tx) - \partial_{i_2}\omega_{t i_3\dots i_{k+1}}(tx) \right]dt\, \times \nonumber\\ \times\,  \partial_{i_1}\theta(x)(1-\theta(x))\, dx^{i_1}\wedge dx^{i_2}\wedge \dots \wedge dx^{i_{k+1}} = 0.
        \end{eqnarray}From this result, a deformed closed $k$-form is exact if and only if the following constraint holds:\begin{equation}\label{eqdiff}\int_0^1 (1-\theta(tx))^k\left[\frac{k}{(1-\theta(tx))}\partial_{i_2}\theta(tx)\omega_{t i_3\dots i_{k+1}}(tx) - \partial_{i_2}\omega_{t i_3\dots i_{k+1}}(tx) \right]dt = 0.\end{equation}Clearly, the expression above admits $\omega = \tilde{d}^2\eta$ as one of its possible solutions, which in terms of components reads\begin{equation}\omega_{i_1 i_2\dots i_k}(x) = -\frac{1}{(1-\theta(x))}\partial_{i_1}\theta(x)\partial_{i_2}\eta_{i_3 i_4\dots i_k}(x).\end{equation}Since condition \eqref{eqdiff} is not satisfied for arbitrary coefficients $\omega_{i_1\dots i_{k}}(x)$, but instead restricts $\omega$ to a (possibly infinite) family of solutions, we conclude that the standard construction of the Poincaré Lemma does not hold for a general deformed $k$-form $\omega$. This suffices to establish the following theorem:\begin{theo}\label{gr}Poincaré's Lemma does not hold for deformed forms on a star-shaped domain in $\widetilde{\mathbb{R}}^n$. That is to say, closed forms on $\widetilde{\mathbb{R}}^n$ are not, in general, exact even locally.\end{theo}
The statement should be understood in the generalized sense appropriate to the deformed differential introduced above. It does not contradict the ordinary Poincaré Lemma, which remains valid for the standard exterior derivative on the same underlying star-shaped domain. For completeness, Theorem \ref{gr} also applies to Scenario 2, in which a form $\omega \in \Lambda^{k}(\widetilde{\mathbb{R}}^n)$ is defined to be closed when $\tilde{d}^2\omega = 0$, and exact when there exists a differential $(k-1)$-form $\eta \in \Lambda^{k-1}(\widetilde{\mathbb{R}}^n)$ such that $\omega = \tilde{d}\eta$. In fact, applying $\tilde{d}$ to $\tilde{d}(H(\mathcal{H}^*\omega)) + H(\tilde{d}(\mathcal{H}^*\omega)) = \omega$, we obtain
\begin{equation}\label{eq:retificacao1}
\tilde{d}^2\left(H(\mathcal{H}^*\omega)\right) + \tilde{d}\left(H(\tilde{d}(\mathcal{H}^*\omega))\right) = \tilde{d}\omega.
\end{equation}
In parallel, applying $\tilde{d}$ to the identity $\tilde{d}(H(\mathcal{H}^*\omega)) + H(\tilde{d}(\mathcal{H}^*\omega)) = i_1^*(\mathcal{H}^*\omega) - i_0^*(\mathcal{H}^*\omega)$ yields
\begin{equation}\label{eq:loucura}
\tilde{d}\left(H(\tilde{d}(\mathcal{H}^*\omega))\right) = i_1^*\tilde{d}(\mathcal{H}^*\omega) - i_0^*\tilde{d}(\mathcal{H}^*\omega),
\end{equation}
under the assumption that $\tilde{d}^2(\mathcal{H}^*\omega) = 0$. Substituting \eqref{eq:loucura} into \eqref{eq:retificacao1}, it follows that $\tilde{d}^2(H(\mathcal{H}^*\omega)) = \tilde{d}\omega - i_1^*\tilde{d}(\mathcal{H}^*\omega) + i_0^*\tilde{d}(\mathcal{H}^*\omega)$, demonstrating that the Poincaré Lemma fails in this case as well.

The explicit construction below is independent of the homotopy argument and provides a direct local obstruction. In particular, it demonstrates the failure of the second notion of the Poincaré property without relying on any global topological information about $\widetilde{\mathbb R}^n$. An explicit example can now be presented. Introducing the operator $\tilde{d}=s\,d,$ 
yields
\begin{equation}
\tilde{d}^{\,2}\omega
=
-s\,d\theta\wedge d\omega.
\end{equation}
Thus, it suffices to construct a form $\omega$ such that $d\theta\wedge d\omega=0$ and $
d\omega\neq0.$ Locally in two dimensions, we can take
\begin{equation}
\theta=x,
\qquad
\omega=y\,dx.
\end{equation}
Then $
d\omega
=
dy\wedge dx
\neq0,$ 
while $
d\theta\wedge d\omega
=
dx\wedge dy\wedge dx
=
0.$ 
Consequently,
 $
\tilde{d}^{\,2}\omega=0.$ 
However, $\omega$ is not necessarily of the form
$\omega=\tilde{d}f=s\,df$,
as would be required by a Poincaré-type lemma. Indeed, a necessary condition for such a representation reads 
\begin{equation}
d\left(\frac{\omega}{s}\right)=0.
\end{equation} Regarding $
s=1-\theta=1-x,$ 
we obtain
$\frac{\omega}{s}=\frac{y}{1-x}\,dx$, 
and hence
\begin{equation}
d\left(\frac{\omega}{s}\right)
=
d\left(\frac{y}{1-x}\,dx\right)
=
\frac{1}{1-x}\,dy\wedge dx
\neq0.
\end{equation}
Therefore, there is no local function $f$ satisfying
$
\omega=\tilde{d}f=s\,df.
$ 
Thus, even though $\tilde{d}^{\,2}\omega=0,$ 
the form $\omega$ is not $\tilde{d}$-exact. 

The results of this section establish three distinct statements. First, the deformed differential defines a higher-order differential structure characterized by third-order nilpotency rather than ordinary second-order nilpotency. Second, this structure naturally gives rise to two generalized notions of cohomology, corresponding to the two possible choices of closed and exact forms. Third, neither generalized notion reproduces the local exactness property of the ordinary Poincaré Lemma in general.

\section{An application to the Hierarchy problem}
\label{sec4}
As it is clear from Section \ref{sec3}, there is no coordinate transformation associated with the form $\tilde{d}x^i$. Indeed, $\tilde{d}^2x^i = -d\theta \wedge dx^i \neq 0$, which highlights the anholonomic nature of the deformed differential forms. The mapping $\tilde{d}x^i = (1-\theta(x))dx^i$ thus represents a frame rescaling rather than a coordinate transformation. Furthermore, the deformed line element (or symmetric $2$-tensor) expressed in this basis is given by\begin{equation}\tilde{\eta}_{ij}\tilde{d}x^i \otimes \tilde{d}x^j = \tilde{\eta}_{ij}(1-\theta(x))^2 dx^i \otimes dx^j.\end{equation}Thus, preserving the algebraic form of the metric requires a compensating conformal transformation $\eta_{ij}\mapsto\tilde{\eta}_{ij} = \frac{1}{(1-\theta)^2}\eta_{ij}$. With these ingredients, one can formulate an approach to the hierarchy problem that parallels considerations in the context of braneworlds and warped geometries \cite{rs,Antoniadis:1998ig}, without requiring extra dimensions. Consider the action for the Higgs field $H$, equipped with a potential allowing for spontaneous symmetry breaking in four dimensions ($n=1+3$), \begin{equation}\label{h1}S \supset \int d^4x\sqrt{-\tilde{\eta}}\{\tilde{\eta}^{ij}(D_iH)^\dagger(D_jH) - \lambda(|H|^2 - v_0^2)^2\},\end{equation}where $D_i$ denotes the covariant derivative. Since $\tilde{\eta}^{ij} = (1-\theta)^2\eta^{ij}$ and  $\sqrt{-\tilde{\eta}} = (1-\theta)^{-4}\sqrt{-\eta}$, Eq. \eqref{h1} becomes\begin{equation}\label{h2}S \supset \int d^4x\sqrt{-\eta}\{(1-\theta)^{-2}\eta^{ij}(D_iH)^\dagger(D_jH) - \lambda (1-\theta)^{-4}(|H|^2 - v_0^2)^2\}.\end{equation}

As discussed in Refs. \cite{jhep,jpa}, physical systems in $\widetilde{\mathbb{R}}^n$ are plagued by Lorentz symmetry violation due to the explicit appearance of $\theta(x)$ terms, as there is no \textit{a priori} reason to attribute the character of a Lorentz scalar to the $\varphi$-deformation. This drawback is bypassed if we study effective models in $\mathbb{R}^{1,3}$. The discussion may be framed in the following way \cite{jpa}: since Eq. \eqref{d1} scales as $\text{length}^{-1}$, its effects are more prominent at small scales or, equivalently, high energies. Thus, one may view $\widetilde{\mathbb{R}}^{1,3}$ as an intermediate spacetime at high energies, whose limit at the highest energy scales evinces a non-trivial fundamental topology while, in the opposite energy direction, recovering $\mathbb{R}^{1,3}$ as the effective background spacetime. Observe that within this perspective there is no change in the fundamental topology, but rather an inability to probe it without a sufficiently high energy scale. In this context, the saturation of $\theta(x)$ at a constant value, say $\bar{\theta}\neq 0$, is particularly appealing for our purposes. Indeed, setting $\theta(x) = \bar{\theta}$ in Eq. \eqref{h2} and performing a field rescaling of $H$, we arrive at\begin{equation}S \supset \int d^4x\sqrt{-\eta}\{\eta^{ij}(D_iH)^\dagger(D_jH) - \lambda(|H|^2 - v^2)^2\},\end{equation}where the symmetry-breaking scale is $v \equiv (1-\bar{\theta})^{-1}v_0$. Thus, the core result is striking: taking $m_0$ at the fundamental Planck scale, every mass $m$ measured in $\mathbb{R}^{1,3}$ is suppressed by a factor $1/(1-\bar{\theta})$. Using the reduced Planck mass $M_{\rm Pl}^{\rm red}\simeq2.4\times10^{18},{\rm GeV}$ and the TeV scale $1,{\rm TeV}=10^3,{\rm GeV}$, we target the hierarchy $
\frac{M_{\rm Pl}^{\rm red}}{1,{\rm TeV}}\simeq2.4\times10^{15}.
$ At this stage, the hierarchy is therefore reproduced by a geometric rescaling of the effective parameters rather than dynamically generated through a solved field equation for $\theta$. The mechanism should consequently be viewed as a kinematical realization of scale separation, with its dynamical completion left open.
Thus,
$
1/(1-\bar{\theta})\sim\frac{1,{\rm TeV}}{M_{\rm Pl}^{\rm red}}\simeq4.2\times10^{-16},
$ 
so that the choice of $1-\bar{\theta}$ is directly tied to the Planck-TeV hierarchy rather than being an arbitrary order-of-magnitude estimate. In this way, TeV scales are recovered if $1-\bar{\theta} \sim 10^{15}$. Such a negative value for $\bar{\theta}$ implies (recalling the discussion around Eq. \eqref{ti}) a saturation value for $\varphi(x)$ given by $\bar{\varphi} \approx \exp(-1.2\times 10^{15})$, likely the smallest number serving a physical purpose. The numerical smallness of this value should not be interpreted as an additional physical input: it is fixed directly by the desired ratio between the electroweak and fundamental scales. The physical question is instead whether a consistent dynamical sector can naturally generate and stabilize such a saturation value. Stated differently, a dynamical mechanism for saturating the $\varphi$-deformation at this precise value must be provided before this approach can be claimed as a complete solution. This is far beyond the scope of the present work. For now, we simply emphasize that the saturation of the $\varphi$-deformation offers a novel pathway to mitigate the hierarchy problem.

\section{Final Remarks and outlook}
\label{sec5}

By incorporating the effects of non-trivial topology and the $\varphi$-deformation into differential forms, we have constructed a deformed differential calculus in which the exterior derivative is no longer nilpotent at second order, but satisfies the higher-order condition $\tilde d^{,3}=0$. This structure leads to a bifurcation of the associated cohomological analysis and, in particular, to the failure of the standard Poincaré Lemma on $\widetilde{\mathbb R}^n$. The resulting obstruction is intrinsic to the deformed differential structure rather than to the ordinary topology of the underlying domain. It persists even on star-shaped regions where the usual de Rham differential admits a contracting homotopy. In this sense, the construction provides a local differential realization of information originating from inequivalent spin structures.

The geometric interpretation of the deformation is encoded in the corresponding anholonomic frame rescaling, which induces a local rescaling of the metric and provides the link between the deformed differential calculus and the physical scales considered in the four-dimensional effective theory. The hierarchy analysis should nevertheless be regarded as an effective realization of this mechanism rather than as a complete dynamical solution. In particular, although a saturation value of the deformation parameter capable of producing the required separation between the fundamental and electroweak scales has been identified, the dynamical origin and stability of this value remain open questions. Establishing such a mechanism would require a more complete theory in which the deformation is itself determined by dynamical equations.

Several directions for future work follow naturally. On the mathematical side, it remains to determine whether at least some version of the generalized cohomological sectors introduced here admits an appropriate de Rham--\v{C}ech correspondence, and to characterize their global and functorial properties. More generally, the relation between the present construction and higher-order differential complexes, including $\mathbb Z_3$-graded structures satisfying $\tilde d^{\,3}=0$, deserves further investigation, particularly in connection with quantum Clifford algebras \cite{Ablamowicz:2014rpa,Goncalves:2023pty,CoronadoVillalobos:2015mns}. On the physical side, the coupling of the deformed differential structure to dynamical gravitational, gauge, and matter fields is an important next step. Black-hole spacetimes provide a particularly natural setting for this analysis, as they may reveal further consequences of the interplay between spin-structure deformations, anholonomy, and higher-order differential structure. These applications will be investigated in forthcoming work \cite{cs}.

\section*{Declaration of generative AI and AI-assisted technologies in the manuscript preparation process}

During the preparation of this work, the authors used Gemini as a language reviewer. After using this tool/service, the authors reviewed and edited the content as needed and took full responsibility for the content of the article.

\section*{Acknowledgments}

JMHS thanks CNPq (Grant No. 307641/2022-8) for financial support. RdR thanks the S\~ao Paulo
Research Foundation – FAPESP (Grants No. 2021/01089-1, No.
2025/23004-9, and No. 2026/14943-4), and to CNPq (Grants No. 303742/2023-2 and No. 401567/2023-0), for
partial financial support. This study was financed in part by the Coordenação de Aperfeiçoamento de Pessoal de Nível Superior - Brasil (CAPES) - Finance Code 001.

\end{document}